\documentclass[journal]{IEEEtran} 

\usepackage{cite}

\usepackage[pdftex]{graphicx}
\graphicspath{{images/}}

\usepackage{amsmath}
\usepackage{mathtools}
\usepackage{amsfonts}
\usepackage{amssymb}
\usepackage{amsthm}
\usepackage{bbm}
\usepackage{algorithm}
\usepackage[noend]{algpseudocode}
\usepackage[font=small,labelfont=bf]{caption}
\usepackage{subcaption}
\usepackage{url}
\usepackage{cleveref}
\usepackage{balance}

\usepackage{color}

\DeclareMathOperator{\Prob}{\mathbb{P}}
\DeclareMathOperator{\E}{\mathbb{E}}

\makeatletter
\def\BState{\State\hskip-\ALG@thistlm}
\makeatother
\algdef{SE}[DOWHILE]{Do}{doWhile}{\algorithmicdo}[1]{\algorithmicwhile\ #1}

\newtheorem{proposition}{Proposition}
\newtheorem{lemma}{Lemma}
\newtheorem{theorem}{Theorem}
\newtheorem{remark}{Remark}
\newtheorem{definition}{Definition}

\usepackage[absolute,showboxes]{textpos}

\TPMargin{5pt}

\newcommand{\copyrightstatement}{
	\begin{textblock}{0.84}(0.08,0.01)    
		\noindent
		\footnotesize
		\copyright 2019 IEEE. This is the authors' version of the article. Personal use of this material is permitted. Permission from IEEE must be obtained for all other uses, in any current or future media, including reprinting/republishing this material for advertising or promotional purposes, creating new collective works, for resale or redistribution to servers or lists, or reuse of any copyrighted component of this work in other works.
	\end{textblock}
}

\begin{document}
\copyrightstatement
\title{Dynamic Binary Countdown for Massive IoT Random Access in Dense 5G Networks\vspace{0.3cm}}

\author{Mikhail~Vilgelm,~\IEEEmembership{Student Member,~IEEE,}  Sergio~Rueda~Li\~nares, and~Wolfgang~Kellerer,~\IEEEmembership{Senior Member,~IEEE} 
\thanks{M. Vilgelm, S. Rueda Li\~nares and W. Kellerer are with the Chair
of Communication Networks, Technical University of Munich, Munich,
Germany. This work was in part supported by the German Research Foundation (DFG) grant KE1863/5-1 as part of the priority program SPP 1914 Cyber-Physical Networking. Corresponding author: M. Vilgelm (mikhail.vilgelm@tum.de).}}

\maketitle

\begin{abstract}
Massive connectivity for Internet of Things applications is expected to challenge the way access reservation protocols are designed in 5G networks. Since the number of devices and their density are envisioned to be orders of magnitude larger, state-of-the-art access reservation, Random Access (RA) procedure, might be a bottleneck for end-to-end delay. This would be especially challenging for burst arrival scenarios: Semi-synchronous triggering of a large number of devices due to a common event (blackout, emergency alarm, etc.). In this article, to improve RA procedure scalability, we propose to combine Binary Countdown Contention Resolution (BCCR) with the state-of-the-art Access Class Barring (ACB). We present a joint analysis of ACB and BCCR and apply a framework for treating RA as a bi-objective optimization, minimizing the resource consumption and maximizing the throughput of the procedure in every contention round. We use this framework to devise dynamic load-adaptive algorithm and simulatively illustrate that the proposed algorithm reduces the burst resolution delay while consuming less resources compared to the state-of-the-art techniques.
\end{abstract}

\begin{IEEEkeywords}
5G NR; RACH; M2M; Random Access; Contention Resolution;
\end{IEEEkeywords}

\section{Introduction}

\IEEEPARstart{T}{he} Internet of Things (IoT) encompasses the set of technologies enabling arbitrary physical objects "to share information and to coordinate decisions" as expressed in~\cite{al2015internet}. IoT involves ubiquitous deployment of sensors, actuators and other computing devices connected to the Internet and placed on any physical object. IoT applications could be enabled by the massive Machine-to-Machine (M2M) communication, the support of which is one of the major design goals for 5G networks~\cite{osseiran2014scenarios}. Massive M2M involves large number of dense sensor network deployments where the key design requirements are low-cost devices, low-energy consumption and wide coverage areas. Examples of such use cases are smart buildings (heating, cooling, etc.), smart cities (parking sensors, public lighting, etc.), logistics and fleet monitoring or smart agriculture (watering, crop monitoring, etc.)~\cite{al2015internet}.

Network designs for massive connectivity inevitably have to address two challenges: to increase capacity and to provide low-overhead connectivity. The challenge of capacity is envisioned to be addressed via increasing network densification (small cells, etc.), increasing spectral efficiency (massive MIMO), and increasing bandwidth. On the other hand, the overhead of auxiliary procedures, especially access reservation, becomes increasingly critical due to the small packet sizes and sporadic (event-driven) transmission patterns of typical IoT applications~\cite{3gpp37.868,al2015internet}. Hence, designing efficient protocols and access reservation procedures is one of the top priority enablers of massive connectivity.

For now, as 5G systems are seen as an evolution of existing LTE-Advanced standard, the access reservation in 5G New Radio (NR) is also based on LTE-Advanced Random Access (RA) procedure. Due to the contention-based nature of Random Access Channel (RACH), which is modeled as multi-channel slotted ALOHA, its scalability becomes a major bottleneck for massive M2M support~\cite{hasan2013random}. Under certain conditions, a significant portion of this massive number of user equipments (UEs) in a given cell may attempt to connect quasi-simultaneously to the next generation eNodeB (gNB), activating themselves over a very short period of time \cite{3gpp37.868}. This scenario is commonly referred to as \textit{burst arrival scenario}, and is likely to occur in emergency situations where a great number of sensors and other IoT devices respond to an unforeseen event, e.g., a fire outbreak activating a large group of sensors or simultaneous attempt to re-connect to the network immediately after power blackout. Burst arrivals have been a known trouble for random access protocols~\cite{firstBatch1988}, and have been confirmed to cause prohibitively long delays for RACH as well~\cite{hasan2013random}.

The standardized way to handle burst arrivals in LTE and NR RA is Access Class Barring (ACB) and its variations~\cite{3gpp37.868}. The idea behind it is to ``smoothen'' the burst by probabilistically delaying the access of some UEs. If the access probability is adjusted dynamically according to the load, ACB becomes a powerful tool for decreasing the delay~\cite{dacb}. However, its normalized throughput is still limited to $\approx1/e$ successful UEs per single Physical RACH (PRACH) preamble. To go beyond that, multiple approaches, which we review in Sec.~\ref{sec:relatedwork}, have been proposed. In our previous work, we have proposed an approach for dense networks, aiding the RA procedure with Binary Countdown Contention Resolution (BCCR)~\cite{PIMRCpaper}. Here, we extend and elaborate on the latter work by developing a practical resource-aware algorithm for joint dynamic optimization of ACB and BCCR.

\subsection{Contributions}
\label{sec:contributions}
In this article, we study the joint operation of ACB and BCCR in dense networks under burst arrival scenario. The novel contributions of the manuscript are:
\begin{itemize}
\item A joint analysis of the RA procedure with both ACB for PRACH preamble transmissions and Binary Countdown Contention Resolution (BCCR) for connection request transmissions. We analyse the expected number of successful UEs in a single contention round, and use it for predicting the expected burst resolution time.
\item Since BCCR introduces an extra overhead, the trade-off between resource consumption and throughput must be studied in detail. We apply a resource-aware framework based on the bi-objective Pareto-optimization, maximizing the expected throughput and minimizing the expected resource consumption.
\item Finally, we use the joint analysis and the resource-aware framework to devise a practical Dynamic Binary Countdown -- Access Class Barring (DBCA) algorithm for resource-aware M2M burst resolution. We simulatively evaluate it against the state-of-the-art for a burst arrival scenario, and demonstrate that DBCA is capable to achieve lower delay while maintaining low resource consumption.
\end{itemize}
\subsection{Structure of the Article}
The remainder of the article is structured as follows. We review the legacy RA procedure with ACB operation and the related works in Sec.~\ref{sec:background}. Next, in Sec.~\ref{sec:recap}, we explain the proposed RA procedure with BCCR. We analyze the protocol under joint ACB and BCCR in Sec.~\ref{sec:analysis}. The analysis is then used to devise a DBCA algorithm in Sec.~\ref{sec:proposal}, and the algorithm's performance is benchmarked against state-of-the-art in Sec.~\ref{sec:evaluation}. Finally, we conclude the article with Sec.~\ref{sec:conclusions}.
\section{Background and Related Work}
\label{sec:background}
\subsection{RA Procedure}
The Random Access procedure is performed in contention rounds. As illustrated in Fig.~\ref{fig:protocol}(a), an arbitrary $i^\text{th}$ round starts with UEs retrieving the system information broadcast from the gNB. The broadcast contains the Physical Random Access Channel (PRACH) configuration, including PRACH resource allocation, Access Class Barring~(ACB)~\textit{access probability} which is denoted by $p_i\in(0,1]$ and the set of available preambles $\mathcal{M}$ for contention-based RA\footnote{Contention-free RA is typically used for handover and is outside of the scope of the present article.} with cardinality $M  = |\mathcal{M}|$.

Before every preamble transmission attempt, the UEs perform an \textit{ACB check} by internally drawing a random number from the uniform $[0,1)$ distribution and only access the PRACH if the number is smaller than $p_i$. Otherwise, they are considered barred and back-off until the next contention round to repeat the process. For the non-barred UEs, the so-called 4-way handshake takes place. It starts with the PRACH preamble transmission (MSG1), in which every UE randomly select one preamble out of the set $\mathcal{M}$ and transmits it on PRACH. If a particular preamble is transmitted by at least one UE, the preamble is denoted as activated. Typically, the gNB cannot distinguish how many UEs have activated a given preamble, so it responds with a Random Access Response (MSG2) to every activated preamble\footnote{In some cases, it is possible to recognize that a collision has occurred, or even to estimate how many UEs are collided~\cite{newCollision}. Using this information, contention resolution can be further improved, but generally we ignore the possibility of collision multiplicity estimation for conservativeness.}. MSG2 informs UEs about the resources allocated in the PUSCH for next RA procedure step, namely, the RRC Connection Request (MSG3). During MSG3 transmissions, the UEs which have previously selected the same preambles for MSG1 will collide, since the gNB is not able to decode multiple RRC requests simultaneously (no capture effect or multipacket reception possibility is assumed). If no collision occurs and the MSG3 is correctly received, the gNB responds with an RRC connection reply (MSG4), and the connection establishment is successfully concluded. In the case of collision, UE does not receive MSG4 and hence assumes that a collision has occurred.

In case of RACH overload, the situation in which more than one UE selects the same preamble becomes increasingly likely, leading to MSG3 collisions and missed Random Access Opportunities (RAOs). Since the number of offered RAOs is limited, RA procedure becomes a bottleneck for the system's throughput under high loads of connecting UEs and leads to substantial resource waste and high access delays.
\begin{figure}[t!]
	\centering
	\includegraphics[width=0.8\linewidth]{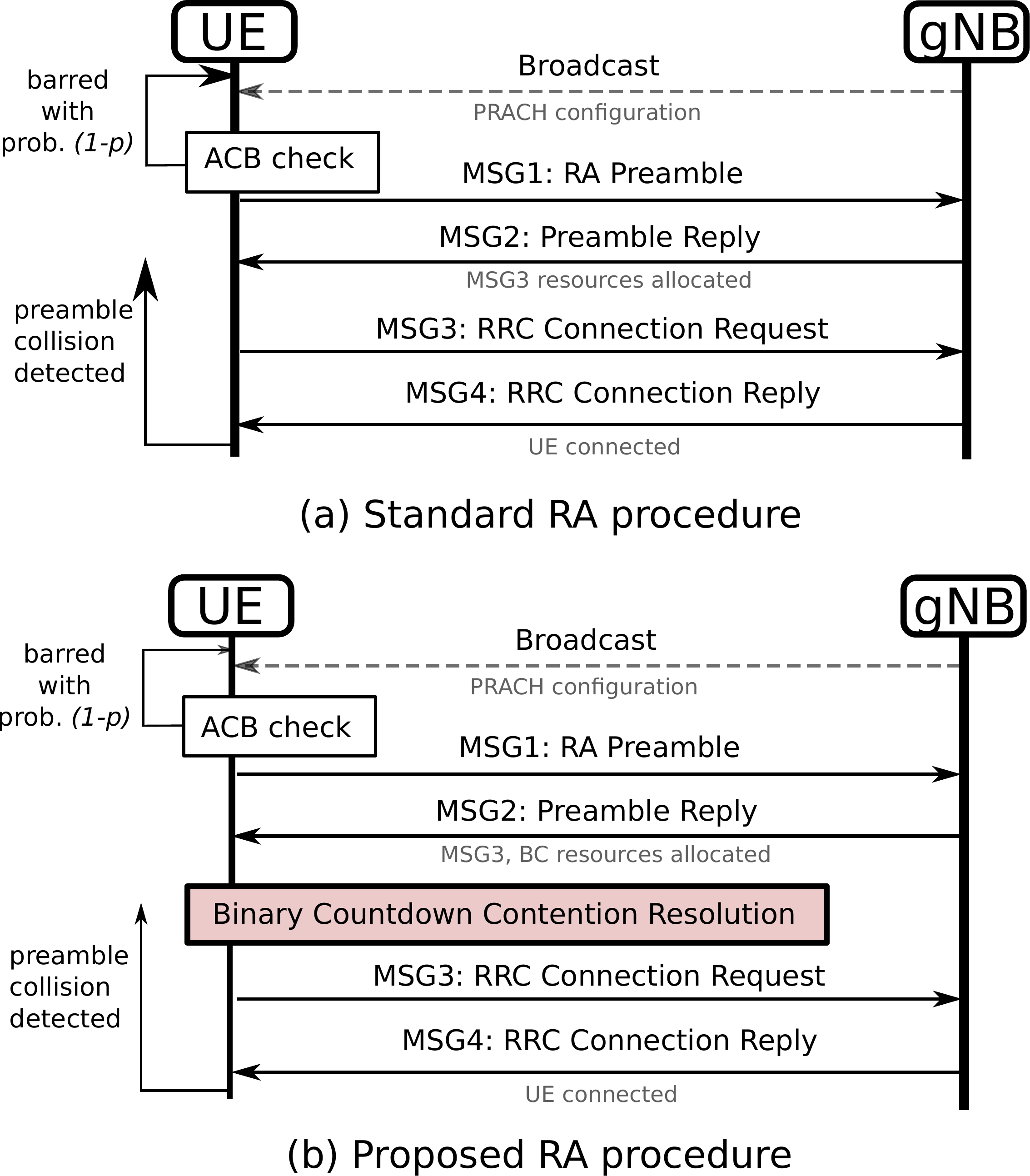}
	\caption{(a) Standard RA procedure with ACB for overload control; (b) Proposed RA procedure, where ACB for MSG1 overload control is combined with BCCR for MSG3 contention resolution.}
	\label{fig:protocol} 
\end{figure}\vspace{-0.2cm}

\subsection{Related Work}
\label{sec:relatedwork}
In this subsection, we review the related work by grouping it into three categories: modeling and analysis of the LTE RA, potential solutions for massive M2M support, and binary countdown for contention resolution in shared medium.

\subsubsection{Modeling and Analysis}

RA procedure is typically modeled as multichannel slotted ALOHA protocol, where the number of preambles available $M$ is considered equivalent to the number of available channels and PRACH periodicity is equivalent to a time slot~\cite{rom2012multiple,andreev2015understanding}. The performance of standardized LTE RA has been extensive studied for two major scenarios: steady-state, typically for independent Poisson arrivals,~\cite{tyagi2015impact,7021909,6786066,vilgelm2017impact,misic2016shout,gharbieh2016tractable}, and transient state for burst arrivals~\cite{wei2015modeling,jian2017random,koseoglu2016lower,cheng2015modeling,8048578,vilgelm2018icc}. Tyagi~\textit{et al.}~\cite{tyagi2015impact} study the back-off based LTE RA without access barring and evaluate the impact of re-transmission limit on the steady-state performance. In the follow-up work, the authors have extended the analysis to RA with ACB~\cite{7021909}. Steady-state M2M random access behavior has also been analyzed in the context of heterogeneous networks with relays in~\cite{6786066} and with intermediate aggregation in~\cite{vilgelm2017impact}. Additionally, related work has been evaluating the effects of power ramping~\cite{misic2016shout} and spacial distribution of the UEs~\cite{gharbieh2016tractable}.

While the above papers have focused on the average steady-state performance, Wei~\textit{et al.}~\cite{wei2015modeling} have presented a drift approximation approach to transient analysis for burst arrival scenarios. The exact probabilistic analysis of LTE RA is elaborated in~\cite{jian2017random}, and the bounds on the average and worst case performance are derived in~\cite{koseoglu2016lower} and~\cite{vilgelm2018icc}, respectively. Additionally, Extended Access Class Barring (EAB) for M2M is analyzed in detail in~\cite{cheng2015modeling}. Detailed performance model of the standardized ACB has been proposed in~\cite{8048578}.

\subsubsection{LTE RA improvements}
In~\cite{3gpp37.868}, 3GPP has proposed a number of potential solutions to the RACH overload problem. Among them are dynamic allocation of RA resources, M2M specific back-off scheme, and ACB. As for dynamic allocation, the authors in~\cite{vilgelm2017latmapa} have proposed an adaptive preamble allocation algorithm, maximizing the per-preamble throughput. It also includes a dynamic allocation of resources for RA, increasing the system's capacity to accept incoming connections whenever an overload in the RACH is detected. Similar to the seminal work on classical slotted ALOHA from Rivest~\cite{rivest1987network}, the authors in~\cite{jin2017recursive,dacb} have proposed a dynamic ACB algorithm, adaptively modifying the access probability according to the traffic load. Achieving similar performance, globally cooperative ACB has been designed in~\cite{6093905}.

The authors in~\cite{ko2012novel} have explored how different UEs choosing the same preamble can be distinguished based on the timing advance. Jang~\textit{et al.}~\cite{jang2014spatial} have introduced a novel spatial group based RA, expanding the available preamble space by exploiting the inverse relation between worst delay profile difference among UEs and the number of distinguishable preambles.
Similarly, Kim~\textit{et al.}~\cite{kim2015enhanced} have proposed a spatial group based RA with reusable preamble allocation, which effectively increases the preamble space if delay profile differences between spatial groups which are allocated the same preambles are assumed larger than multi-path delay spread. In~\cite{pratas2016random}, the employment of frames composed of a number of successive PRACH slots has been investigated. The authors in~\cite{wiriaatmadja2015hybrid} have proposed a hybrid protocol which combines RA and payload transmission, leveraging the fact that M2M devices are likely to have low payload volume. Random access with multi-user detection for multiple-antenna OFDMA has been developed in~\cite{7823034}.

Another class of algorithms revisited from the early years of slotted ALOHA research are the \textit{tree resolution algorithms (TRA)}~\cite{capetanakis1977multiple,tsybakov1978free}. The authors in~\cite{madueno2014efficient} have studied the use of a $q$-ary TRA in LTE RA, by reserving groups of $q$ preambles in a given PRACH for every collision. Their proposal was found out to resolve $30000$ synchronous arrivals needing only an average of 5 preamble transmissions per device. In~\cite{gursu2017hybrid}, two hybrid collision avoidance-tree algorithms have been proposed for LTE, combining TRA with pre-back-off and achieving a per-preamble throughput of up to 0.4295. 

\subsubsection{Binary Countdown for Contention Resolution}

The BCCR protocol is known primarily from the CAN bus systems, but it has been also employed in powerline communication~\cite{gehrsitz2014priority}, and studied academically for ad-hoc networks~\cite{you2003new,hashiura2010performance}. Recently, BCCR has been revisited as a possible option for access reservation in the next generation networks~\cite{baiocchi2017random,sen2011no}. In our previous work, we have independently proposed a simple framework for integrating BCCR into RACH~\cite{PIMRCpaper}. Here, we extend it by considering a joint optimization of preamble contention (via ACB) and MSG3 contention (via BCCR).
\vspace{-0.2cm}
\section{Random Access Procedure with Binary Countdown Contention Resolution}
\label{sec:recap}
As we see from the previous section, the majority of the RA improvements proposed in the state-of-the-art focus on the \textit{preamble contention} step. However, we note that the actual collision, although being a direct consequence of the preamble collision, is occurring at MSG3 transmission step. Following this observation, our approach aims at resolving a MSG3 contention while allowing the preamble collision. For that, we invoke the Binary Countdown Contention Resolution (BCCR) protocol prior to MSG3 contention~\cite{PIMRCpaper}. Addition of BCCR makes our approach largely orthogonal to the state-of-the-art, as the preamble contention could be optimized independently of the MSG3 contention. 

In this section, we explain the basics of BCCR protocol (\ref{sec:countdown}), its integration into 5G NR (\ref{sec:integration}), and give an illustrative example of RA operation with BCCR (\ref{sec:example}).

\subsection{Recap: Binary Countdown Protocol}
\label{sec:countdown}
\begin{figure*}[t!]
	\centering
	\includegraphics[width=0.75\textwidth]{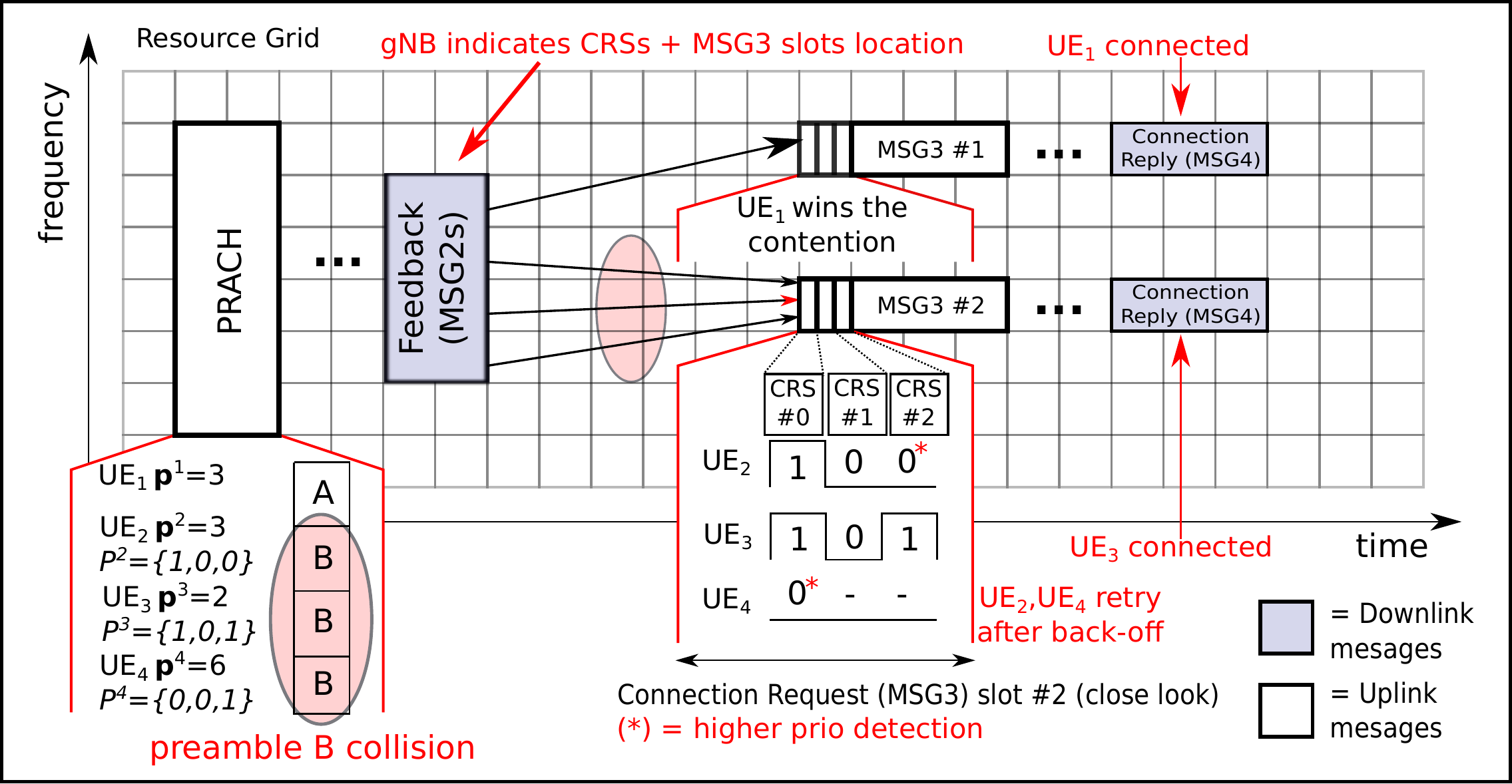}
	\caption{Exemplary operation of RA procedure with Binary Countdown Contention Resolution shown for two activated preambles A and B. For every activated preamble, additional resources are allocated for $k_i=3$ CRSs prior to MSG3 transmission. CRSs remain unused in the case of preamble A, whereas CRSs for preamble B the CRSs are used to resolve the contention.}
	\label{fig:procedure}
 	\vspace{-0.5cm}
\end{figure*}

The core idea of BCCR is to use short contention resolution slots (CRSs) prior to the packet transmission, to probabilistically ``decide'' in a distributed fashion which of the contending UEs transmits the packet.

To explain the protocol, we denote the number of CRSs in a given contention round $i$ as $k_i$, and the number of associated \textit{priority levels} $l_i\triangleq2^{k_i}$. Before the start of the BCCR procedure, each contending UE $x$ uniformly at random chooses a priority level $\mathbf{p}^{(x)}$. The selected level $\mathbf{p}^{(x)}$ is represented as a $k_i$-digit binary sequence $P^{(x)} = \left[ p_0^{(x)}, \ldots ,p_j^{(x)}, \ldots, p_{k_i-1}^{(x)} \right]$, corresponding to the base-2 representation of $\left(l_i - 1- \mathbf{p}^{(x)} \right)$, where $p_j^{(x)} \in \{ 0,1\}$ and $0\leq j \leq k_i-1$. As an example, if we set $k_i = 2$, we have the highest priority level $\mathbf{p}_{\max,i} = 0$ represented by $P_{\max,i} = \left[ 1, 1 \right]$, while the lowest priority $\mathbf{p}_{\min,i} = 3$ is represented by $P_{\min,i} = \left[ 0, 0 \right]$. Here, we follow the convention that \textit{$0$ is the highest priority}.

The binary sequence, generated from the chosen priority level, is then used by the UE to decide its behavior in any CRS $\#j$. Starting from the CRS $\#0$ onward, a contending UE $x$ is either listening to the medium if $p_j^{(x)} = 0$, or transmitting a signal to inform other contenders of its presence if $p_j^{(x)} = 1$. If, in any CRS, a silent UE detects another UE transmitting, it assumes there is a contending UE with higher priority and \textit{immediately} abandons the contention, i.e., it does not transmit in any later CRS regardless of its priority. If, on the contrary, a UE completes the $k_i$ CRSs without having detected any UE with higher priority, it assumes that it is the winner of the contention and proceeds to send its packet.

Priority levels can be assigned in a number of different ways, as further elaborated in~\cite{PIMRCpaper}. For an analogous problem, the authors in~\cite{baiocchi2017random} have already proven the uniform distribution to be the success-maximizing random priority distribution. For that reason, randomized priority selection with uniform distribution is assumed for the rest of this article.
\vspace{-0.2cm}
\subsection{BCCR Integration in 5G NR}
\label{sec:integration}
In contrast to bus or Wi--Fi systems, contention in LTE and NR starts with sending a random PRACH preamble, which makes the BCCR not applicable on the first step. However, as the actual collision occurs at the step three (MSG3), we proposed to allocate PUSCH resources for BCCR prior to MSG3, hence, extending the MSG3 slot by $k_i$ CRSs~\cite{PIMRCpaper}. Thus, the resulting RA procedure combines of two techniques: overload control prior to preamble transmission by the means of ACB, and contention resolution prior to MSG3 transmission using BCCR (see the time-frequency grid illustration in Fig.~\ref{fig:procedure}).

The duration of a CRS has to take into account the granularity of resource allocation, required resources for MSG3 duration, switching time between RX and TX. These factors are mostly limited by the technology standard. While for LTE the allocation granularity is conservative and limited to 1 sub-frame, in 5G NR smaller and more flexible CRS configurations are possible due to flexible frame structure and finer scheduling granularity down to 1 OFDM symbol~\cite{dahlman20185g}. To stay inline with NR scheduling, we assume a CRS to consume 1 Resource Block (RB) bandwidth x 1 symbol period per CRS basis, so that $T_{\text{CRS}} = 1$~OFDM symbol.

It is important to note that, to have the full effect, BCCR requires that all contending UEs are overhearing each other's broadcast signals. Since most of the M2M use cases target dense networks and correlated arrivals, \textit{small to medium geographical size events} can fully exploit BCCR, while for larger events there might be performance penalties due to hidden terminal problem. Additionally, short CRS duration implies that time alignment problems due to propagation delay are possible. We have investigated it in our previous work~\cite{PIMRCpaper}, showing that UEs should not be farther than roughly $1$~km apart, to have no time alignment problem. The exact size of fully supported events depends on the UE distribution and placement, network density, etc., and is outside the scope of the article.
\vspace{-0.2cm}
\subsection{Example: RA with BCCR}
\label{sec:example}
Fig.~\ref{fig:procedure} shows an example of BCCR operation with $k_i=3$ CRSs. After the preamble transmission is received, gNB allocates the resources for BCCR and MSG3 for every activated preamble and informs UEs about the allocated CRSs and their position in the time-frequency grid by MSG2 feedback. In the case of preamble A, it has only been activated by UE$_1$, so there is no collision to be avoided. Note, however, that UE$_1$ still needs to perform BCCR prior to sending MSG3, since the number of UEs occupying a certain preamble is unknown. Although unused CRSs introduce extra overhead, we will show in the later sections that this overhead is negligible compared to the gains of BCCR in high-load regime. Moreover, this overhead can be avoided if gNB can distinguish a collided from a singleton preamble during the first step of RA procedure~\cite{newCollision}. 

In contrast, UE$_2$, UE$_3$, and UE$_4$ have all activated the same preamble B. They then perform BCCR with randomly chosen priorities $P^{(2)} = \left[1,0,0 \right]$, $P^{(3)} = \left[1,0,1 \right]$ and $P^{(4)} = \left[0,0,1 \right]$. UE$_2$ and UE$_3$ transmit a signal in CRS $\#0$, which is sensed by the listening UE$_4$. Thus, UE$_4$ immediately abandons the contention and does not participate in any further CRSs, regardless of its priority. In CRS $\#1$, both UE$_2$ and UE$_3$ remain silent and listen to the medium, and both detect no transmission. Finally, in CRS $\#2$, UE$_2$ remains silent while UE$_3$ transmits a signal. Thus, UE$_2$ also abandons the contention, leaving UE$_3$ as sole winner, which then proceeds to sending MSG3 without collisions and to successfully connect with the gNB. In this case, BCCR has avoided what would otherwise have been a wasted RAO, turning it into a successful connection.
\section{Joint Binary Countdown - Access Barring Performance Analysis}
\label{sec:analysis}
In this section, we analyse the performance of the joint access barring and binary countdown operation. First, the system model is described in~\ref{sec:systemmodel}, then we derive an expected throughput in a single contention round in~\ref{sec:expanalysis}, extend it towards bi-objective optimization problem in~\ref{sec:paretooptimality}, and finally generalize for the full burst resolution delay in~\ref{sec:delayanalysis}.

\subsection{System Model}
\label{sec:systemmodel}

We consider a burst arrival scenario as proposed in~\cite{3gpp37.868}. $N$ UEs in a cell with one gNB are semi-synchronously activated. At time $t<0$, all UEs are disconnected from the gNB. During the interval $0\leq t < T_a$, every UE commences the connection procedure at a random time $t$ with probability distribution $g_a(t)$. The probability distribution is representing an arrival process, with three main possibilities: beta-distributed, uniformly random, and simultaneous ``spike'' arrivals with $T_a=0$. 
\begin{definition}[PRACH slot]
We denote the periodicity of PRACH in the resource grid as a~\textit{PRACH slot}, or~\textit{slot}.
\end{definition}

It has been shown that, in the current networks, the collision feedback as well as broadcast periodicity might exceed PRACH slot duration~\cite{8048578, wei2015modeling}. Although the signaling procedures are more flexible in NR standard, the definition of a \textit{contention round} is necessary in order to generalize the analysis accounting for different possible RACH implementations.

\begin{definition}[Contention round]
We denote the minimum period within which the contention parameters can be adjusted and the collision feedback can be received as a contention round. A single contention round can comprise one or multiple PRACH slots.
\end{definition}

Prior to any contention round $i$, every UE undergoes an ACB check: With the access probability $p_i$ it proceeds to contend, and with probability $1-p_i$ it skips the upcoming round\footnote{In other words, ACB represents a geometric random back-off. It is one possible back-off option, and it can serve as an approximation for other back-off schemes or combinations thereof.}. If the ACB check is passed, UE chooses (uniformly random) a $j^\text{th}$ preamble, with $j \in \{1, ..., M\}$, where $M$ is the total number of preambles available in a given contention round. Prior to BCCR, each preamble can have one of three possible outcomes: \textbf{idle} if no device occupies the preamble; \textbf{successful} if one and only one device chooses the preamble; and \textbf{collided} otherwise.

For any collided preamble $j$, at most one UE among those having chosen it can be successfully resolved via BCCR. For every available preamble, we have one RAO associated to it. Random variable representing the number of UEs choosing a given preamble $j$ as $m_{i,j}$, and outcome of a RAO as $x_{i,j}$, we define the resulting collision channel model as:
\begin{equation}
x_{i,j} \triangleq \begin{cases}
1 & m_{i,j}=1\text{ }\cup\text{ }\left(m_{i,j}>1\text{ }\cap\text{ }\text{resolved via BCCR}\right),\\
0 &\text{otherwise.}
\end{cases}
\label{eqn:channelmodel}
\end{equation}

The contention in a collided preamble ($m_{i,j}>1$) is defined as \textit{resolved} via BCCR, if one of the UEs has uniquely chosen the highest priority among the set of the priorities selected by UEs occupying preamble $j$.

As the focus of this article is on the interplay between BCCR and ACB, we make an additional assumption that the downlink channel resources are sufficient and do not pose a performance bottleneck.

\subsection{Single Contention Round: Expected Throughput}
\label{sec:expanalysis}
To analyze RACH under joint effect of ACB and BCCR, consider the system state prior to a contention round $i$. We denote the number of competing UEs at this point as $n_i$. Since we assume $p_i$-persistent ACB with no drops, $n_i$ accounts both for backlogged users and newly arrived ones, as there is no distinction between their behavior.

\begin{theorem}
\label{theo:expsuccess} Given $n_i$ competing UEs, access probability $p_i$, and $l_i$ BCCR priority levels, the expected number of successful RAOs $S\triangleq\E\left[s_i\right]$ \emph{(throughput)} in the contention round $i$ is:
\begin{equation} \label{eq:exp_S}
S = \frac{n_ip_i}{l_i} \sum_{h=1}^{l_i} \left(1 - \frac{h}{l_i} \frac{p_i}{M} \right)^{n_i-1}.
\end{equation}
\end{theorem}
\begin{proof}
	By the definition~\eqref{eqn:channelmodel}, the outcome of an arbitrary RAO $j$ is successful, $x_{i,j}=1$, if a unique UE chooses the RAO, or if it it wins a contention by choosing the highest priority level. For an arbitrary priority level $\mathbf{p}^\prime = h-1$ with $h \in \{1,2,\ldots, l\} $, there are $h$ levels with equal or higher priority. E.g., $ \mathbf{p}^\prime \triangleq 0 \Rightarrow h = 1 $ (1 higher or equally prioritized level), or $ \mathbf{p}^\prime = l_i - 1 \Rightarrow h = l_i $. Consider a UE that has passed the ACB check, has chosen preamble $j$ and has chosen a priority $h-1$, we obtain its successful BCCR probability:
	\begin{eqnarray}
	\Prob\left[x_{i,j}=1 | \mathbf{p}^\prime, \text{ACB passed}\right] = \left(1 - \frac{h}{l_i} \frac{p_i}{M} \right)^{n_i-1},
	\end{eqnarray}
	where $\frac{h}{l_i} \frac{p_i}{M}$ represents the probability that another UE passes ACB, chooses preamble $j$ and higher or equal priority level. Since the events of choosing any priority levels are a partition of the sample space, we conclude:
	\begin{align} 
	&\Prob\left[x_{i,j}=1|\text{ACB passed}\right]  =  \nonumber\\
	&=\sum_{h=1}^{l_i} \Prob\left[\mathbf{p}^\prime=h-1\right]\Prob\left[x_{i,j}=1|\mathbf{p}^\prime = h-1, \text{ACB passed}\right] \nonumber\\
	&=\sum_{h=1}^{l_i} \frac{1}{l_i} \left(1 - \frac{h}{l_i} \frac{p_i}{M} \right)^{n_i-1},\label{eqn:theo1temp}
	\end{align}
	where the summation over $h\in\{1,...,l_i\}$ considers any possible priority level.
	By analogy, accounting for the probability of a UE to pass ACB check $p_i$, choose the preamble $j$, and that any of $n_i$ could be successful, we obtain a modified expression~\eqref{eqn:theo1temp}:
	\begin{equation} \label{eq:exp_X1}
	\Prob[x_{i,j} = 1] = 
	\binom{n_i}{1}\frac{p_i} {l_iM}\sum_{h=1}^{l_i}\left( 1 - \frac{h}{l_i} \frac{p_i}{M} \right)^{n_i-1}.
	\end{equation}
	By definition of expectation, we get:
	\begin{equation}
	\E[x_{i,j}] = \sum_{t=0}^{1}t\Prob[x_{i,j}]=\Prob[x_{i,j}=1].
	\label{eqn:exp_succ_preamble}
	\end{equation}
	
	To obtain \eqref{eq:exp_S}, we recall that $s_i=\sum_{j=1}^{M}x_{i,j}$, and use the sum of the expectations rule:
	\begin{equation} \label{eq:exp_sum}
	S = M\E\left[  x_{i,j} \right] = \frac{n_ip_i}{l_i} \sum_{h=1}^{l_i} \left(1 - \frac{h}{l_i} \frac{p_i}{M} \right)^{n_i-1}.
	\end{equation}
\end{proof}

The implications of the theorem are illustrated in Fig.~\ref{fig:ana_vs_sim}, where the expected throughput is plotted as a function of $p_i$ for different values of  $l_i$, for a fixed $M=54$~preambles and $n_i=1000$~UEs. Clearly, increasing $l_i$ improves the throughput, and increases the supported load, by shifting the peak of the curve to the right. We also observe that the analytical results are closely matching the simulation.

\begin{figure}[t!]
	\centering
	\includegraphics[width=\linewidth]{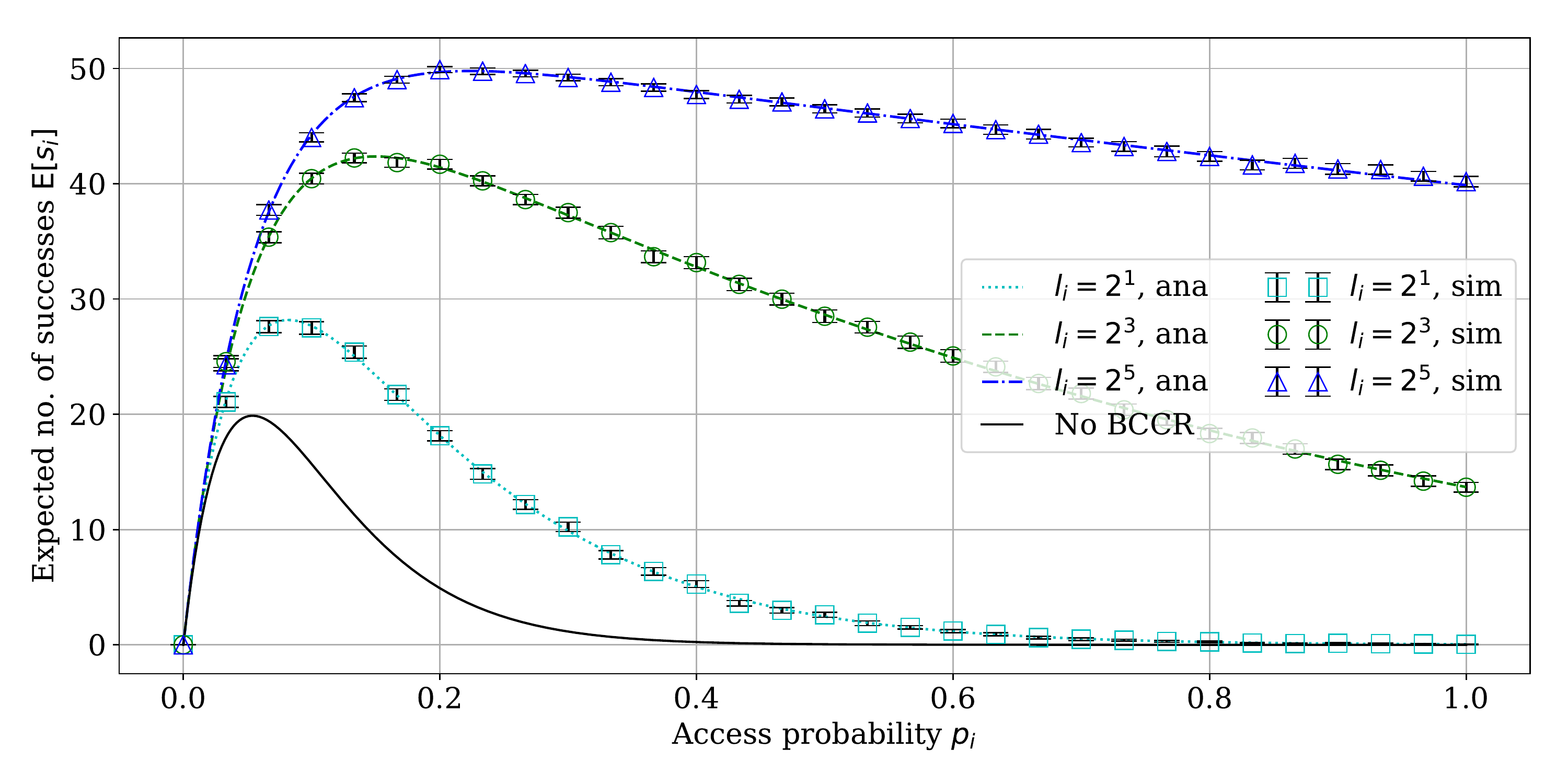}
	\caption{Expected throughput $S\triangleq\E\left[s_i\right]$ vs. access probability for different number of priority levels in a single contention round, $M=54$, $n_i=1000$ UEs, with $0.95$ confidence intervals.}
	\label{fig:ana_vs_sim}
\end{figure}

\subsection{Single Contention Round: Bi-objective Optimization}
\label{sec:paretooptimality}
When applying ACB in the RACH, the \textit{access probability} $p_i$ must be chosen and broadcast by the gNB prior to every contention round $i$. Adding BCCR introduces a new design parameter into the problem, namely the number of CRSs $k_i$. Its value must also be chosen by the gNB and communicated to the UEs along with the $p_i$, so that it is known by all the participants prior to MSG3 transmissions. For a given $n_i$, we define a pair of values $(p_i,k_i)$ as an \textit{operating point}. 

In the state-of-the-art, RACH problem is typically approached as a maximization of the expected throughput $S$. With access probability $p_i$ being the only design parameter (no BCCR), there is a single optimal point~\cite{jin2017recursive,dacb} 
\begin{equation}
p_i^\star=\min\left(1,\frac{M}{n_i}\right).\label{eqn:pstar}
\end{equation} 

However, this approach is not applicable to our modified procedure. It is clear that, for $n_i > 1$, increasing $k_i$ always has a positive effect on throughput, and it is intuitively clear that BCCR can achieve an arbitrary small collision probability. On the other hand, adding CRSs introduces an overhead as every CRS consumes additional time-frequency resources. Hence, we face a fundamental trade-off between two competing optimization goals: maximizing the expected throughput and minimizing the expected resource consumption.

To assess this trade-off, we need to quantify the number of resources consumed per contention round. For the sake of simplicity, we focus only on the consumed resources in the uplink channels. First, to characterize the $i^\text{th}$ contention round's outcome, we introduce the auxiliary variables denoting number of idle and occupied (successful or collided) preambles, respectively: $M_i^I\triangleq\sum_{j=1}^{M}\mathbbm{1}_{m_{i,j}=0}$, $M_i^O\triangleq\sum_{j=1}^{M}\mathbbm{1}_{m_{i,j}\geq1}$, where $\mathbbm{1}_X$ is the indicator function of a subset defined by condition $X$. Recall that $m_{i,j}$ denotes the number of UEs occupying a preamble $j$.
\begin{lemma}
\label{lem:expnonidle}
Given $n_i$ backlogged UEs, and access probability $p_i$, the expected number of occupied preambles $\E\left[M_i^O\right]$ in the $i^\text{th}$ contention round is:
\begin{equation} \label{eq:exp_O}
\E\left[M_i^O\right] =  M - M \left( 1 - \frac{p_i}{M} \right)^{n_i}.
\end{equation}
\end{lemma}
\begin{proof}
	Following an analogous approach as for Theorem~\ref{theo:expsuccess}, we consider a single preamble $j$ first. Denote by $y_{i,j}\triangleq\mathbbm{1}_{m_{i,j}\geq1}$ the binary random variable indicating occupation of the preamble $j$ in the round $i$. The probability that a given preamble is idle is obtained then as
$
	\Prob[y_{i,j}=0] = \left(1 - \frac{p_i}{M} \right)^{n_i}.
$
	Proceeding analogously as in Eqn.~\eqref{eq:exp_sum}, we obtain~\eqref{eq:exp_O} as:
	\begin{equation} \label{eq:exp_Y1}
	\E[M_i^O]=M\sum_{j}\E\left[y_{i,j}\right] 
	= M - M\left(1 - \frac{p_i}{M} \right)^{n_i}.
	\end{equation}
\end{proof}

Clearly, the number of occupied preambles is increasing with increasing access probability $p_i$. According to the procedure, for every occupied (activated) preamble, resources for MSG3 + $k_i$ CRSs transmissions are allocated.

\begin{proposition}
	\label{prop:resourceconsumption}
	For a given contention round $i$, the expected uplink resource consumption, as a function of the number of contending UEs $n_i$, is:
\begin{equation}
R = R_1 + r_3\left(1 + k_i\delta \right)\underbrace{\left(M - M\left(1-\frac{p_i}{M}\right)^{n_i}\right)}_{\text{expected occupied preambles}}\emph{~RBs}.\label{eqn:res_consumption}
\end{equation}
where $R_1$ are the resources consumed by PRACH, $r_3$ the resources consumed by every MSG3 transmission and $\delta\triangleq \frac{r_{\text{CRS}}}{r_3}$ is the relative overhead introduced by each CRS with respect to $r_3$, i.e., a CRS consumes $\delta r_3$ resources.
\end{proposition}
\begin{proof}
Follows directly from Lemma~\ref{lem:expnonidle}.
\end{proof}
It is clear that $R_1$ is deterministic and only depends on the number of allocated preambles and their format. The second summand is stochastic, and depends on the preamble occupation.

We now have all the necessary elements to formulate the trade-off between resource consumption and throughput as a bi-objective optimization problem:
\begin{subequations}
\label{eq:opt_problem}
\begin{eqnarray}
&&\underset{p_i,k_i}{\min} \quad\left\lbrace -S, R\right\rbrace\\
\text{with }&&S =  \frac{n_ip_i}{l_i}\sum_{h=1}^{l_i} \left( 1 - \frac{h}{l_i}\frac{p_i}{M} \right)^{n_i-1}\nonumber\\
			&&R =R_1 + r_3(1 + k_i \delta) \left(M - M\left(1-\frac{p_i}{M}\right)^{n_i}\right)\nonumber\\
\text{s.t.} &&p_i\in(0,1],\\
&& k_i \in \mathbb{Z}_{\geq 0}.
\end{eqnarray}
\end{subequations}

As it is a multi-objective optimization problem, we study Pareto optimal points, that is, solutions for which there is no other possible solution which simultaneously performs better with respect to one of the optimization goals without degrading the other. These points constitute the Pareto frontier. Introducing BCCR into the RA procedure dramatically modifies how the structure of Pareto frontier looks like, but it does not modify the problem's dual nature of conflicting objectives.
\begin{figure}[t!]
	\centering
	\begin{subfigure}[t]{\linewidth}
		\includegraphics[width=\linewidth]{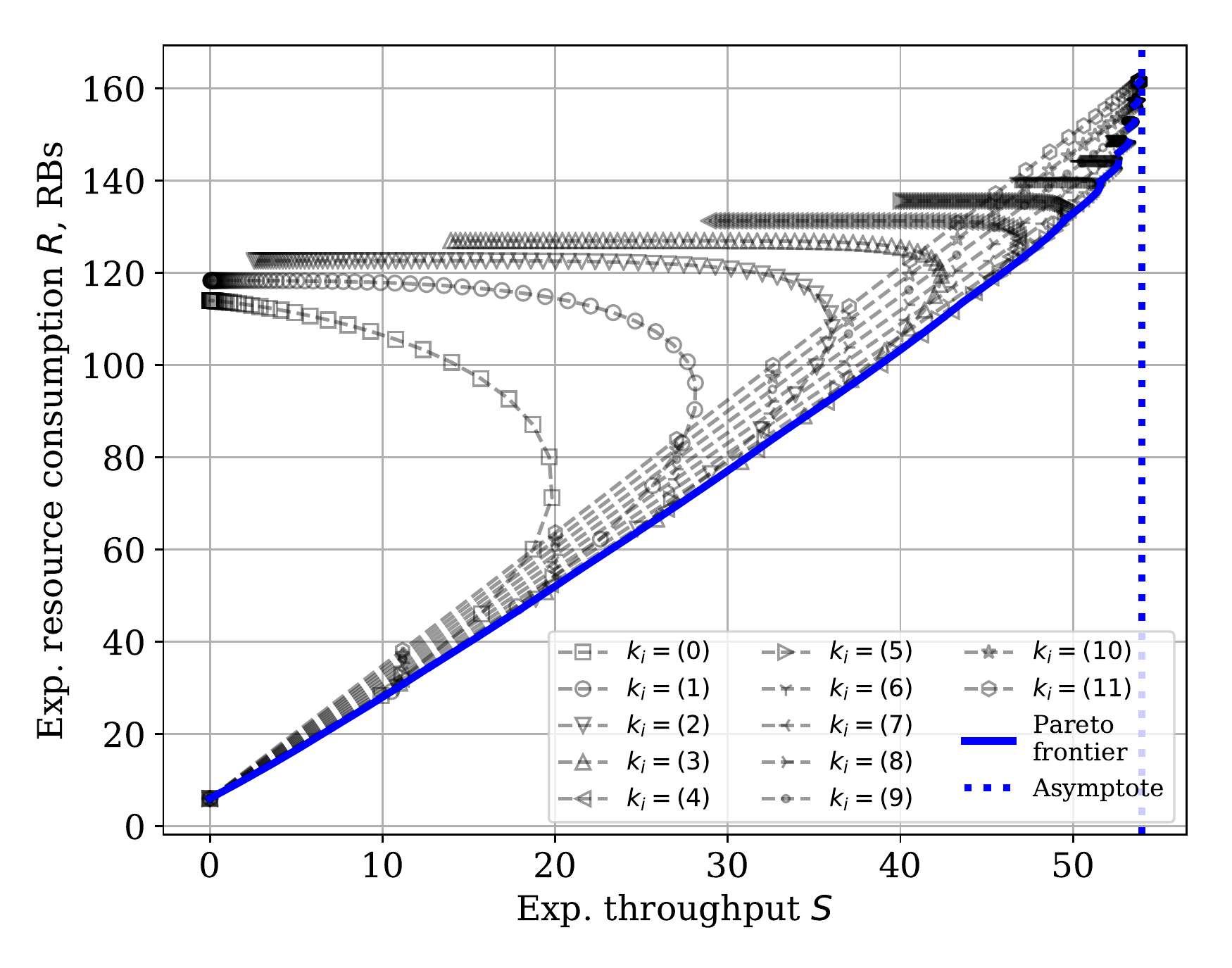}
	\vspace{-0.7cm}
	\end{subfigure}
	\caption{Expected consumed resources vs. expected number of successes, showing the Pareto frontier. $R_1=6$~RBs, $r_3=2$~RBs, $\delta=0.07$, $M=54$ preambles, $n_i=1000$ UEs.}
	\label{fig:paretofrontier}
\end{figure}

Exemplary Pareto frontier produced numerically for the optimization problem defined by~\eqref{eq:opt_problem} is plotted in Fig.~\ref{fig:paretofrontier} for $n_i = 1000$, assuming that MSG3 takes $r_3=2$ RBs (1 sub-frame in time domain), and a CRS occupies 1 OFDM symbol, i.e., $\delta=\frac{1\text{ symbol}}{1\text{ sub-frame}} \approx 0.07$. Every black (with markers) curve corresponds to achievable performance for a fixed value of $k_i$, and varying the values of $p_i$. We observe that the Pareto frontier is a combination of achievable performance curves for different values of $k_i$. The ratio of $S/R$ is almost constant through a large part of the Pareto frontier, and starts decreasing as throughput approaches the total number of preambles $M$. This means that for large values of $k_i$, increasing it further results only in marginal improvement of throughput.

We further observe that the Pareto frontier exhibits an asymptotic behavior at:
\begin{equation}
\label{eq:pareto-asymptote}
\lim\limits_{k_i\to+\infty}S = \E[M_i^O|p_i=1]=M-M \left( 1-  1/M\right)^{n_i}.
\end{equation}

This follows since the expected throughput is constrained by the maximum expected number of occupied preambles, while the expected resource consumption, in our simplified model, is not constrained at all.
It readily follows that:
\begin{equation}
\lim\limits_{n_i\to+\infty}\E[M_i^O|p_i=1]=M.
\end{equation}
Hence, we can asymptotically achieve normalized per-preamble throughput of $1$. This result coincides with analytical studies of other binary countdown-based protocols, showing that arbitrary small collision probability could be achieved~\cite{baiocchi2017random}.

We will return to the bi-objective optimization and design a practical burst resolution algorithm obtaining a Pareto-optimal solution in the next Sec.~\ref{sec:proposal}.

\subsection{Full Burst: Expected Resolution Time}
\label{sec:delayanalysis}

To generalize the single contention round analysis towards the full burst resolution time $T_{\text{BR}}$, we apply a modified \textit{drift approximation model} proposed by Wei~\textit{et al.}~\cite{wei2015modeling}. Since describing the exact evolution of the backlog over time is a tedious task, the authors~\cite{wei2015modeling} propose to approximate it by considering only the evolution of the expectation of the backlog. Let us introduce additional notation of new arrivals during the contention round $i$ as $a_{i}$. The backlog state at any round is thus represented by the following random variable recursion:
\begin{equation}
n_{i+1} = n_i - \underbrace{s_i}_{\text{successful UEs}} + \underbrace{a_{i}.}_{\text{new arrivals}}
\end{equation}
Now, to compute the expected burst resolution time, we can consider this recursion in the expectation:
\begin{equation}
\E[n_{i+1}] = \E[n_i] - \E\left[s_i\right] + \E[a_{i}].
\end{equation}
Expected throughput $\E[s_i]=S$ in a given round is computed via Theorem~\ref{theo:expsuccess}, and the expected arrivals in a contention round, dependent on the arrival process, are computed via the probability density function of the activation time $g_a(t)$ as:
\begin{equation}
\E[a_{i}] = N\int_{(i-1)T_{\text{PRACH}}}^{iT_{\text{PRACH}}}g_a(t)dt.
\end{equation}
Now, computing the expected burst resolution time $\E\left[T_{\text{BR}}\right]$ with an arbitrary precision $\epsilon$ simplifies to an iterative process with a stopping condition:
\begin{equation}
\E\left[T_{\text{BR}}\right] = t,\quad\text{if }\E[n_{t}]<\epsilon\text{ and } \E[a_i]=0\text{ }\forall\text{ }i\geq t.
\end{equation}
\section{Implementation: Dynamic Binary Countdown - Access Barring (DBCA)}
\label{sec:proposal}

In this section, applying the analytical results and observations from the previous sections for the practical design of RA procedure, we propose a Dynamic Binary Countdown - Access barring (DBCA) protocol. In the core of the protocol is the idea to dynamically determine at every contention round $i$ the values of $p_i$ and $k_i$ from the Pareto frontier. To make the protocol more practical, we also aid it with a backlog estimator, since backlog is unknown to the gNB in most of the scenarios.

DBCA protocol consists of the following main steps, repeated in every contention round:
\begin{enumerate}
\item[\textbf{I.}] Contenting UEs undergo ACB and (if successful) transmit MSG1s
\item[\textbf{II.}] gNB receives MSG1s and updates the estimates of the number of contending UEs $\hat{n}_{i}$.
\item[\textbf{III.}] Based on the estimate $\hat{n}_i$, gNB calculates the number of CRSs $k_i$ to be used for MSG3 transmissions and informs contending UEs about it as part of the MSG2.
\item[\textbf{IV.}] UEs undergo BCCR and (if successful) transmit MSG3s
\item[\textbf{V.}] gNB receives MSG3s and updates the estimate of the number of backlogged UEs $\hat{n}^-_{i+1}$ for the next round.
\item[\textbf{VI.}] Based on the estimate $\hat{n}^-_{i+1}$, gNB calculates barring factor $p_{i+1}$ for the next round, and informs UEs via system information broadcast.
\end{enumerate}

The pseudocode for gNB-side of DBCA is presented in Algorithm~\ref{alg:proposed_algorithm}. In the following, we explain in the choice of the operating point $(p_i,k_i)$ given the backlog estimate in steps III and VI, the estimation procedure in steps II and V.

\subsection{Choosing the operating point $(p_i,k_i)$}
The core of the DBCA algorithm, corresponding to steps III and VI of the algorithm, is choosing the operating point on the Pareto curve. To obtain a Pareto-optimal solution, the respective bi-objective optimization problem is solved by scalarization, where we convert both objectives into one using the preferences of a decision maker. We apply scalarization by $\epsilon$-constraint method~\cite{miettinen2008introduction}: A constraint is set on one objective function, and the system of optimized for the second objective. Either the minimum desired throughput or the maximum allowed resource consumption can be used as a constraint. Typically, however, the resource constraint $R \leq \epsilon_R$ would be a major limiting factor. In this case, the optimization problem targets the maximization of the expected throughput subject to the constraint on the expected resource consumption. We formulate it as follows:
\begin{subequations}
\label{eq:p_solver}
\begin{eqnarray}
\max_{p_{i},k_{i}}&&\text{ }S\left(k_{i},p_{i};n_{i} \right) \\
\text{s.t.} &&R \left(k_{i},p_{i};n_{i} \right) \leq \epsilon_R\\
&&k_{i} \in \mathbb{Z}_{+},\:p_{i}\in (0,1]
\end{eqnarray}
\end{subequations}
\begin{remark}[On the Pareto optimality]
	$\epsilon$-constrained method ensures at least weak Pareto optimality~\cite{miettinen2008introduction}. If multiple optimal solutions to the problem~\eqref{eq:p_solver} are found, strong Pareto optimality with respect to~\eqref{eq:opt_problem} can be enforced by choosing the solution with the lowest resource consumption. 
\end{remark}
The operating point choice is split into two stages, since $p_i$ and $k_i$ must be allocated at different times: $k_i$ prior to MSG3, and $p_i$ prior to MSG1. First, consider the $k_i$ allocation at stage II. To maximize the expected number of successes, gNB observes the outcome of the preamble transmission (number of activated preambles $r_i$), and decides $k_i$ according to the estimated $\hat{n}_i$, and subject to the expected resource consumption constraint $\epsilon_R$. 
Setting $R=\epsilon_R$, and solving Eqn.~\eqref{eqn:res_consumption} for $k_i$ and rounding to the nearest integer, we obtain the decision rule for the number of CRSs:
\begin{equation}
\label{eq:k_solver}
k_i = \left \lceil\frac{1}{\delta}\left(\frac{\epsilon_R - R_1} {r_3 \left(M - M\left(1-p_i/M\right)^{\hat{n}_i}\right)} - 1 \right)\right\rfloor.
\end{equation}
Then, $k_i$ is communicated to the UEs as a part of the MSG2 alongside with the uplink grants for the MSG3 transmission. 
\begin{remark} Note that we are considering \textit{soft constraints}, which apply only to the expectations. This constraints are useful from the system design and dimensioning perspective, allowing the gNB to manage the resource distribution between different functions or slices. Hard constraint can be enforced by substituting the term $M\left(1-p_i/M\right)^{\hat{n}_i}$ in Eqn.~\eqref{eq:k_solver}, representing the expected number of idle preambles, with the observed value $M_i^I$.
\end{remark}
Later, upon completion of the contention round (line 11 of the pseudocode), gNB observes the number of successful outcomes $s_i$ and updates the backlog estimation. At this moment, $p_{i+1}$ for the next cycle is decided, as the solution to the problem~\eqref{eq:p_solver}, where we use a priori backlog estimate $\hat{n}_{i+1}^-$ for $n_i$. This solution can be found numerically, and we will return to the complexity of the solution in Sec.~\ref{sec:complexity}.
This \textit{access probability} is then broadcast before the $(i+1)^{\text{th}}$ contention round.

\begin{algorithm}[t!]
{\footnotesize
	\caption{Pseudocode for Dynamic Binary Countdown - Access barring (DBCA): gNB View.}\label{alg:proposed_algorithm}
	\begin{algorithmic}[1]
		\State Initialize $i=0$, $\hat{n}^-_0 = 1$, $p_0 = 1$, $q_0 = 0$
		\For {every contention round $i$}
		\State Observe $M_i^I$\Comment{stage \textbf{II}}
		\State Compute $\Delta \hat{n}_i$ via Eqn.~\eqref{eq:delta_v_pseudoBayesian}
		\State $\hat{n}_i=\hat{n}_i^-+\Delta \hat{n}_i$\Comment{update a posteriori backlog estimate}
		\State Compute $k_i$ via Eqn.~\eqref{eq:k_solver}\Comment{stage \textbf{III}}
		\State Allocate resources for $k_i$ CRSs and MSG3s
		\If {$\Delta\hat{n}_i>0$}\Comment{stage \textbf{V}}
		\State $q_{i+1}=q_i+1$ \Comment{correction for bursty arrivals}
		\Else { $q_{i+1}=0$}
		\EndIf
		\State Observe successful MSG3 transmissions $s_i$
		\State $\hat{n}^-_{i+1} = \hat{n}^-_{i}+q_{i+1}\Delta\hat{n} - s_i$ \Comment{update a priori backlog estimate}
		\State Compute $p_{i+1}$ via Eqn.~\eqref{eq:p_solver} \Comment{stage \textbf{VI}}
		\EndFor
	\end{algorithmic}}
\end{algorithm}

\subsection{Estimating the Backlog $\hat{n}_i$}
In most of the practical cases, the size of the backlog at any time step $n_i$ is unknown to the gNB. Hence, we have to adapt the procedure in order to obtain an estimate of the backlog $\hat{n}_i$. There exist multiple state-of-the-art estimation techniques, all relying on the observation of each contention round outcomes, i.e., number of idle $M_i^{I}$ and occupied $M_i^{O}$ preambles. In this work, we adapt the pseudo-bayesian estimation from~\cite{jin2017recursive} to the joint binary-countdown access barring procedure.

The estimation of the backlog is reflected at two points in the algorithm: to decide the number of contention resolution slots (stage I) after observing the number of idle preambles $M_i^I$ (note that at this moment the number of successful UEs is unknown); and to decide the access probability for the $(i+1)^\text{th}$ contention round, after the number of successful UEs $s_i$ is known (stage V).

First, let us consider stage I. It calculates the \textit{a posteriori} estimation $\hat{n}_i$ as a function of the \textit{a priori} estimate $\hat{n}_i^-$ (which depends on the previous RA round estimation, hence its recursiveness) and the number of idle preambles $M_i^I$. It assumes the backlog size in the $i^\text{th}$ contention round is a Poisson random variable whose mean is the \textit{a priori} estimate $\hat{n}^-_i$ and calculates the correction~\cite{jin2017recursive}:
\begin{equation}
\label{eq:delta_v_pseudoBayesian}
\Delta \hat{n} =  p_i\hat{n}^-_i \left( e^{-\frac{p_i\hat{n}^-_i}{M}} - \frac{M_i^I}{M} \right) \left(1 - e^{-\frac{p_i\hat{n}^-_i}{M}} \right)^{-1},
\end{equation}
The \textit{a priori} estimation is then corrected:
\begin{equation}
\label{eq:correction_pseudoBayesian}
\hat{n}_t = \hat{n}^-_t + \Delta \hat{n}
\end{equation}

The stage V starts once the results of the complete $i^{\text{th}}$ RA contention round are obtained. A simple \textit{a priori} estimate for the next contention round $i+1$ is computed as in \cite{jin2017recursive}:
\begin{equation}
\label{eq:update_v}
\hat{n}^-_{i+1} = \hat{n}_{i} + \alpha_{i+1}^- - s_i
\end{equation}
where $\alpha_{i+1}^-$ is an \textit{a priori} estimation of the \textit{arrivals} during the next round. As we assume that no information about the arrivals distribution is available, we take $\alpha_{t+1}^-$ proportionally to the number of arrivals in the previous RA round and estimate it as $\alpha_{i} = \max(0, \Delta \hat{n}_i)$. As the estimation we use is an adaptation of the Enhanced Pseudo-Bayesian ACB algorithm from~\cite{jin2017recursive}, we also use a heuristic involving a ``boosting factor'' $q_{t+1}$ in the \textit{a priori} estimation to better adjust for the burst arrivals:
\begin{equation}
\label{eq:arrivals_estimator}
\alpha^-_{i+1} = q_{i+1}\cdot\alpha_{i} = q_{i+1}\cdot\max(0, \Delta \hat{n}_i)
\end{equation}

\subsection{Complexity}
\label{sec:complexity}
Clearly, the algorithm complexity is dominated by the line 13 of the pseudocode, where $p_{i+1}$ is computed as a solution to the problem~\eqref{eq:p_solver}. While this is a non-linear mixed integer optimization, it is possible to find a solution efficiently, considering that if we fix $k_i$, the resulting problem of finding optimal $p_i^\star$ has a unique solution. We state the second fact as Lemma~\ref{lemma:fixed_k_problem}.
\begin{lemma}
\label{lemma:fixed_k_problem}
Given fixed number of CRSs $k_i=\bar{k}$, and the number of UEs $n_{i}=\bar{n}>2$, the optimization problem
\begin{equation}
\label{eq:p_solver_fixed_k}
\begin{aligned}
\underset{p_{i}}{\text{\emph{max} }}S,\quad
\emph{s.t. } &R \leq \epsilon_R,\quad\:p_{i}\in (0,1],
\end{aligned}
\end{equation}
has a unique solution $p_i^\star$.
\end{lemma}
\begin{proof}
See Appendix~\ref{app:lemma2}.
\end{proof}
The Lemma~\ref{lemma:fixed_k_problem} implies that for any fixed $k_i$, we can find optimal $p_i^\star$ fast with any numerical methods or local search, e.g., gradient descent. To further simplify the problem, we convert it to a root finding problem in Lemma~\ref{lemma:rootfinding}.
\begin{lemma}
	\label{lemma:rootfinding}
The problem defined by Eqn.~\eqref{eq:p_solver_fixed_k} can be equivalently solved by 
\begin{equation}
p^\star_i=\min\left(\frac{x^\star Ml_i}{n_i},p_{\max}\right),
\end{equation}
where $p_{\max}$ is given by~\eqref{eqn:const_pmax} and $x^\star$ found either as a root of
\begin{equation}
\left(1-x\right)+e^{-xl_i}\left(\left(1-xl_i\right)\left(e^{-x}-1\right)+x\right)-e^{-x} = 0,
\label{eqn:rootfinding}
\end{equation}
or as $x^\star=\frac{n_i}{Ml_i}$ if no roots exists for $x\in\left(0,\frac{n_i}{Ml_i}\right]$.
\end{lemma}
\begin{proof}
See Appendix~\ref{app:lemma3}.
\end{proof}
Practically, a na\"{\i}ve Python implementation according to Lemma~\ref{lemma:rootfinding} based on scipy.optimize package~\cite{scipy} yields $20-35$~$\mu$s average execution time. Using the fact the objective function is increasing in $k_i$ and in any realistic implementation $k_i$ is upper-bounded by $k_{\max}$\footnote{In general, resource grid and resource management might impose different granularity constraints on the allocation of CRS. As this constraints are implementation specific and hard to model realistically, in this work we assume that $k_i$ could be allocated with granularity $1$, i.e, any number of slots up to $k_{\max}$ can be allocated.}, the optimal operating point $\left(p_i^\star,k_i^\star\right)$ can be found with a search over $[0,k_{\max}]$. Hence, the worst-case complexity of the step is upper bounded by $\mathcal{O}(k_{\max})$.
\begin{table}[t!]
	\renewcommand{\arraystretch}{1.3}
	\caption{Summary of simulation parameters.}
	\label{tab:simparams}
	\centering
	\begin{tabular}{|l||l|}
		\hline
		\hline
		Contention round / PRACH slot & $10$~ms\\
		\hline
		Preambles per slot $M$ & 54\\
		\hline
		Number of UEs $N$ & $500-10000$\\
		\hline
		Act. time $T_a$: uniform, beta & $1$~s / $100$~c.rounds\\
		\hline
		Act. time $T_a$: delta & 1~c.round\\
		\hline
		Beta distribution parameters $(\alpha,\beta)$ & (3,4)~\cite{3gpp37.868}\\
		\hline
		Resource constraint proportionality constant $C$ & $1.0-1.8$\\
		\hline
		Maximum number of CRSs $k_{\max}$ & 14 \\
		\hline
		CRS allocation granularity & 1\\
		\hline
		Resources per PRACH channel & $6$~RBs\\
		\hline
		Resources per MSG3 $r_3$ & $2$~RBs~\cite{jang2017non}\\
		\hline
		Single CRS relative overhead $\delta=r_{\text{CRS}}/r_3$& 0.07\\
		\hline
		CRS duration $t_{\text{CRS}}$ & 1 OFDM symbol\\
		\hline
	\end{tabular}
\end{table}

\begin{figure*}[t!]
	\centering
	\includegraphics[width=\textwidth]{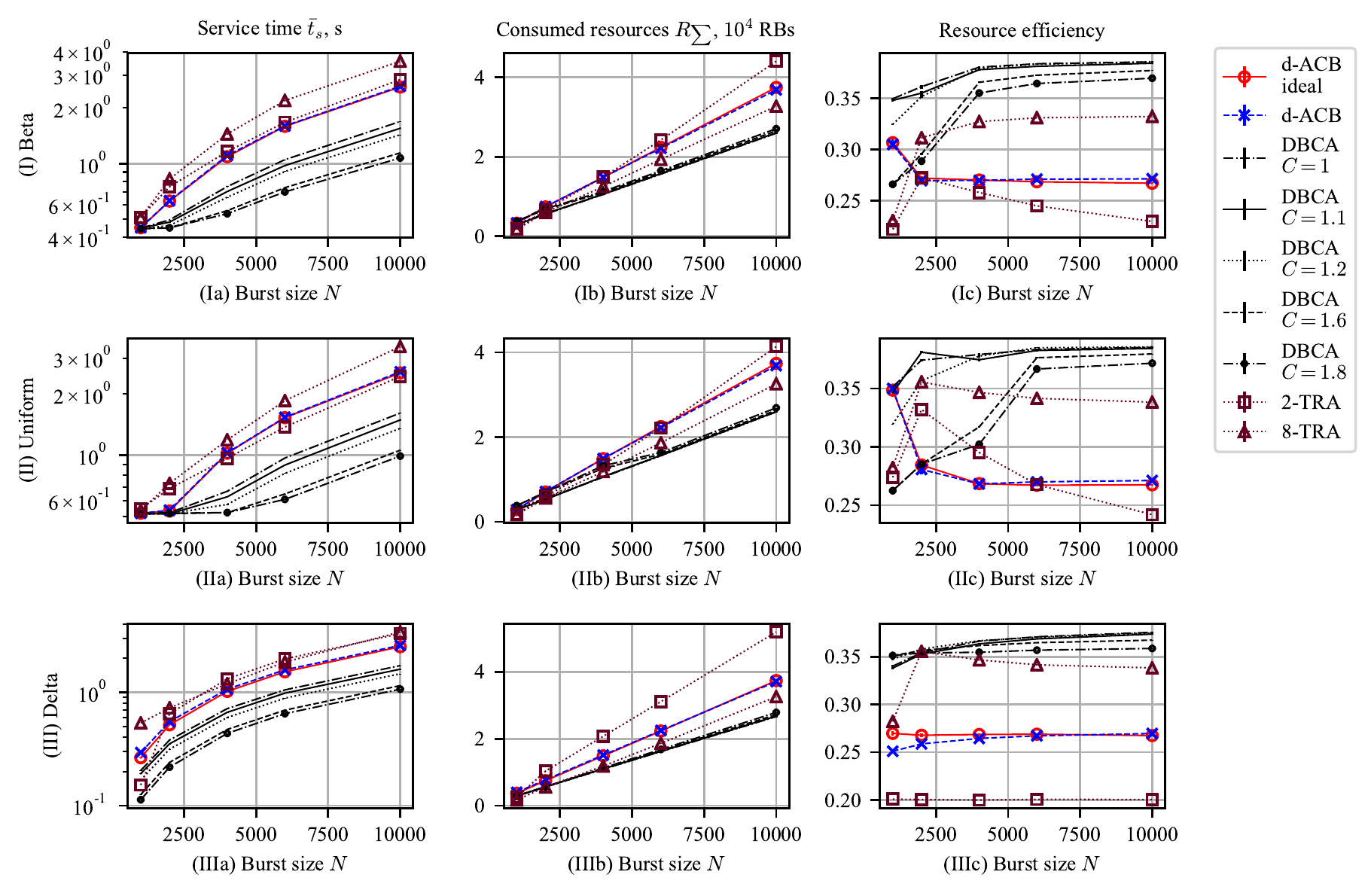}
	\caption{DBCA algorithm performance for different values of $C$ vs. d-ACB and $\{2,8\}$-ary TRA for three arrival scenarios: (I) Beta, (II) Uniform, and (III) Delta. The plots (Ia-IIIa) show mean service time $\bar{t}_s$ of UEs in a burst, (Ib-IIIb) average total consumed resources $R_{\sum}$ for a burst resolution, and plots (Ic-IIIc) show the resource efficiency $U$. The $x$-axis represents the burst size $N$. The $95$~\% confidence intervals do not exceed 1.1~\% of the mean. They are thus omitted to avoid visual clutter.}
	\label{fig:proposedperformance}
	\vspace{-0.3cm}
\end{figure*}

\section{Evaluation}
\label{sec:evaluation}
In this section, we present simulative evaluation of the performance of DBCA and compare it to the baseline of Dynamic Access Class Barring (d-ACB)~\cite{dacb} and $q$-ary Tree Resolution Algorithm ($q$-TRA)~\cite{madueno2014efficient} by the means of a custom event-based simulator. We present the simulation set-up (\ref{subsec:sim}), metrics (\ref{subsec:metric}), and finally the results (\ref{subsec:results}).

\subsection{Simulation Set-up}
\label{subsec:sim}

We simulate a burst arrival scenario, with three burst arrival distributions~\cite{3gpp37.868}: delta, uniform activation time distribution, and beta arrivals:
\begin{equation}
g_a(t)=\begin{cases}
\frac{t^{\alpha-1} (T_a-t)^{\beta-1} }{T_a^{\alpha+\beta-2} B (\alpha,\beta)}, & 0\leq t< T_a\quad\text{(if beta)},\\
\frac{1}{T_a} & 0\leq t< T_a\quad\text{(if uniform)},\\
1 & t=0\quad\text{(if delta)},\\
0 & \text{otherwise}.
\end{cases}
\label{eqn:arrivalprocess}
\end{equation}
where $B(\alpha,\beta)$ denotes the Beta function.

We simulate the four-way handshake of RA procedure using the collision channel model as defined by~\eqref{eqn:channelmodel}, hence, our simulation only captures MAC layer effects. The simulation is organized in contention rounds, where UEs are assumed to receive MSG4 if successful. If no MSG4 is received by the end of the contention round, UE assumes a collision. For simplicity, we assume that a contention round takes one PRACH slot, and a PRACH slot length is assumed to be equal to $10$~ms, which corresponds to one PRACH allocation per frame (e.g., configuration index 5 in LTE or 18 in NR)~\cite{3gpp36.211,3gpp38.211}. Considering a more practical model would potentially have a quantitative effect on the results, but it is left aside in order to obtain more illustrative performance evaluation. We choose an exemplary value of $k_{\max}=14$ by assuming that the amount of resources spent on BCCR is at most equal to the resources spent on MSG3, i.e., $k_{\max}\delta=R_3$, and the duration of one CRS equals to one OFDM symbol. The simulation parameters are summarized in the Table~\ref{tab:simparams}. We present average values obtained from at least 30 Monte-Carlo simulations for each data point, with $95$~\% confidence intervals not exceeding $1.1$~\% of the mean.
\begin{remark}
		We only evaluate our approach for the case of one burst without any background traffic and assume that the amount of background traffic is negligible, as is common in the literature, e.g.,~\cite{dacb,gursu2017hybrid,wei2015modeling}. Although DBCA is not optimized for the presence of background traffic, the estimation steps II and V would implicitly take it into account by over-estimating the number of back-logged UEs involved in a contention. Alternatively, if a localized burst arrival is detected or anticipated (e.g., during group paging), some preambles could be reserved specifically for the burst resolution and advertised in the system broadcast respectively~\cite{vilgelm2017latmapa}.
\end{remark}

\subsection{Performance Metrics}
\label{subsec:metric}
Three performance metrics are investigated: mean service time $\bar{t}_s\triangleq\sum_{j=1}^{N}t_s^j / N$ (time until a UE successfully completes RA procedure), mean consumed uplink resources $R_{\sum}$ throughout the whole burst resolution duration, and mean resource efficiency: successful outcomes, normalized by the consumed uplink resources, defined as $U\triangleq\frac{1}{T_{\text{BR}}}\sum_{i\in T_{\text{BR}}}s_i/\left(R_1+r_3(1+k_i\delta)M_i^O\right)$.

To provide a fair comparison, the per-contention period resource constraint is set proportional to the expected resource consumption of a d-ACB algorithm~\cite{dacb} under the same conditions:
$
\label{eq:res_const}
\epsilon_R \triangleq C \times R \left(\hat{n}_i, p^\star_{i}, k_i=0 \right),
$
where $C$ is the proportionality constant and $p_i^\star$ is the access probability maximizing the expected success rate for RA without BCCR as defined by~\eqref{eqn:pstar}. Intuitive meaning of proportionality constraint is the following. The case of $C=1$ corresponds to the case where the proposed algorithm DBCA cannot consume more resources per contention round than the baseline d-ACB, so it provides a fair comparison. The case of $C>1$ studies how can DBCA benefits from the additional resources, which d-ACB cannot make use of. It is not straightforward to enforce resource constraint on the TRA, therefore we simulate TRA without resource constraint, giving it an advantage. We chose to simulate TRA with branching factors $q\in\{2,8\}$. For illustrative purposes, we also include an ideal version of d-ACB in the evaluation, where a perfect knowledge of the state information is assumed.

\subsection{Results}
\label{subsec:results}

In Fig.~\ref{fig:proposedperformance}, we see how the proposed algorithm performs for different values of the proportionality constant $C$ compared to the baseline algorithms. From Figs.~\ref{fig:proposedperformance}(Ia-IIIa), we observe that DBCA provides lower average service time that the baseline for most of the arrival distributions. Only in the case of uniform arrivals with low load $N\leq2000$ UEs, DBCA performs similar to d-ACB with small service times for all algorithms. Overall, service times grow almost linearly with $N$ in high load regime for all arrival distributions and algorithms, whereas a non-linear behavior is noted for low-to-medium load $N\leq 4000$ UEs for uniform and beta arrivals. $q$-ary TRA with high branching factor $q=8$ does not provide significant advantage, while low branching $q=2$ performs well for uniform arrivals, however, still worse than DBCA for medium-to-high load.

In Figs.~\ref{fig:proposedperformance}(Ib-IIIb), we observe the relationship between total resource consumption and load (burst size). First of all, we note dominantly linear growth of resource consumption with the load, with less steep slopes for DBCA protocols. As a result, DBCA consumes similar amount of resources as the baselines in low load regimen, and significantly less in high load regimen. Even in the case where $C = 1$, where DBCA is constrained to consume no more resources than d-ACB per contention-round, DBCA still yields lower overall consumption for full burst resolution.
This follows since DBCA is capable of obtaining more throughput out of the same consumed resources as d-ACB, thus, wasting less resources to collisions. 
Interestingly, binary TRA consumes consistently more than other baselines, while $8$-ary -- consistently less. This is a counter-intuitive observation, however, it is easily explained: While higher branching factor provides sub-optimal throughput, it produces mostly idle preambles, and idle preambles consume significantly less resources than collided. This result is confirmed in the Figs.~\ref{fig:proposedperformance}(Ic-IIIc), where we see that $8$-ary TRA is very resource efficient. Overall, we observe that DBCA with low $C$ performs most efficient with $U\geq 0.35$ for medium-to-high load.

Another counter-intuitive observation is that the overall resource consumption of DBCA exhibits relatively low variation for the different values of resource constraint $CR$, especially for higher values of $N$. However, there is indeed a great difference in the mean service time, which is the lower the higher $C$ is. This is explained by the fact that, if we look into Fig.~\ref{fig:paretofrontier}, we see that the ratio $S/R$ is relatively steady along most of the Pareto frontier for high values of $n_i$, which is the condition where most of the resource consumption takes place. Thus, we can trade consuming higher amounts of resources for a shorter period of time, or lower amounts for a longer period of time, without severely affecting the resulting resource efficiency.

\section{Conclusion}
\label{sec:conclusions}

\subsection{Summary}
In this article, we are addressing the challenge of providing massive IoT support for dense 5G networks. We propose to aid 5G Random Access procedure with Binary Countdown Contention Resolution, allowing to resolve the contention prior to RRC Connection Request transmissions. We analyse the performance of joint ACB and BCCR operation. We further apply a framework for resource-aware RA optimization, and, based on it, propose a Dynamic Binary Countdown -- Access class barring (DBCA) for fast and resource efficient burst resolution. DBCA is benchmarked via an event-based simulation against other state-of-the-art solutions, such as dynamic ACB and $q$-ary tree resolution algorithms for different burst arrival processes. Our simulation results confirm that DBCA allows to connect bursts of UEs faster while also consuming less resources. The advantages of DBCA are especially prominent when large bursts of UEs occur.

\subsection{Future Work}
Since BCCR relies on all UEs listening to each other, its performance is at its best in highly dense networks, and this is the scenario we have targeted in this article. Future work could address the scenarios of partial overhearing, there the UEs are not always close to each other, and access BCCR gains for such scenarios. The impact of other propagation scenarios, e.g., near-far conditions and capture effects, is also an open question to be addressed by the future work. Furthermore, we have shown here that DBCA is resource efficient. In general, it can be expected that resource efficiency implies reduced power consumption, but further study is needed on how to extend the problem formulation to account for the power consumption of UEs~\cite{wang2018multi}. In this work, we have targeted 5G networks in sub-$6$~GHz spectrum, thus neglecting usage of mmWave and beam forming procedures. Applications of our approach to random access in mmWave could be an interesting challenge for future work~\cite{7010532}. Finally, as we show in Sec.~\ref{sec:paretooptimality}, BCCR asymptotically allows to achieve arbitrary low collision probability. This property could be utilized in for ultra-reliable low latency (URLLC) applications, to design RA procedure with reliability guarantees~\cite{vilgelm2018icc}.

\appendices
\section{Proof of Lemma~\ref{lemma:fixed_k_problem}}
\label{app:lemma2}
First, we prove that the unconstrained problem has only one solution. Consider the objective function as a product of two functions $S=f_o(p_i;\bar{k})\triangleq y_o(p_i)g_o(p_i)$, where $y_o(p_i)\triangleq\frac{\bar{n}}{l_i}p_i$, and $g_o(p_i)\triangleq\sum_{h=1}^{l_i}\left( 1 - \frac{h}{l_i} \frac{p_i}{M} \right)^{\bar{n}-1}$. The first and second order derivatives of these functions are:
\begin{align}
&\frac{dy_o}{dp_i}=\frac{\bar{n}}{l_i},\quad\frac{d^2y_o}{dp_i^2}=0,\label{eqn:yo_deriv}\\
&\frac{dg_o}{dp_i} = -\frac{(\bar{n}-1)}{l_iM}\sum_{h=1}^{l_i}h\left(1-\frac{hp_i}{l_iM}\right)^{\bar{n}-2},\label{eqn:go_deriv1}\\
&\frac{d^2g_o}{dp_i^2} = \frac{(\bar{n}-1)(\bar{n}-2)}{l_i^2M^2}\sum_{h=1}^{l_i}h^2\left(1-\frac{hp_i}{l_iM}\right)^{\bar{n}-3}.\label{eqn:go_deriv2}
\end{align}

Note that since the following holds: $\frac{dg_o}{dp_i}<0$, $\frac{d^2g_o}{dp_i^2}>0$, $g_o(p_i)$ is a convex and strictly decreasing function. Now we prove by contradiction that the function $f_o(p_i)$ has a single maximum. Assume that $f_o(p_i)$ has two maximums $p_{i,1}$ and $p_{i,3}$, that implies there is also has a minimum in $p_{i,2}$. Considering that $\frac{d(y_og_o)(p_{i,j})}{dp_i}=0$, $j\in\{1,2,3\}$, and Eqns.~\labelcref{eqn:yo_deriv,eqn:go_deriv1,eqn:go_deriv2}, we obtain:
\begin{align}
&\frac{dg_o(p_{i,j})}{dp_i}=-\frac{1}{p_{i,j}}g_o(p_{i,j})\label{eqn:go_first}\\
&\frac{d^2g_o(p_{i,j})}{dp_i^2}=\frac{g_o(p_{i,j})}{p_{i,j}}\left(\frac{1}{p_{i,j}}-1\right)\label{eqn:go_second}.
\end{align}

Using Eqns.~\eqref{eqn:go_first},~\eqref{eqn:go_second} we can derive the second derivative of the function $(y_og_o)(p_{i,j})$ as:
\begin{equation}
\frac{d^2(y_og_o)(p_{i,j})}{dp_i^2}=\underbrace{-\frac{\bar{n}}{l_i}}_{<0}g_o(p_{i,j})\underbrace{\left(\frac{1}{p_{i,j}}+1\right)}_{>0}.
\end{equation}
Following our assumption, we have $\frac{d^2(y_og_o)(p_{i,j})}{dp_i^2}<0$, $j\in\{1,3\}$, and $\frac{d^2(y_og_o)(p_{i,j})}{dp_i^2}>0$, $j=2$. This implies that $g_o(p_{i,1}),g_o(p_{i,3})>0$, and $g_o(p_{i,2})<0$. However, as $g_o(p_i)$ is a decreasing function, and $p_{i,1}<p_{i,2}<p_{i,3}$, we come to a contradiction. Hence, $f_o(p_i)$ has only one maximum, and the unconstrained problem~\eqref{eq:p_solver_fixed_k} has only one solution.

Next, consider the constraint function:
\begin{equation}
R \left(p_{i};\bar{n},\bar{k}\right) = R_1 + r_3(1+\bar{k}\delta)\left(M-M\left(1-\frac{p_i}{M}\right)^{\bar{n}}\right)\leq \epsilon.
\end{equation}
We can reformulate it as: 
\begin{align}
p_i \leq p_{\max},\label{eqn:const_pmax}\text{ with }p_{\max}=M - M\left(1-\frac{\epsilon-r_3}{Mr_3(1-\bar{k}\delta)}\right)^{\frac{1}{\bar{n}}}.
\end{align}
Hence, this constraint is a closed half-plane defined by a constant $p_{\max}$. Clearly, since $f_o(p_i)$ is a function, it has only one interception with~\eqref{eqn:const_pmax}. This implies that the constrained problem~\eqref{eq:p_solver} also has also only one solution.

\section{Proof of Lemma~\ref{lemma:rootfinding}}
\label{app:lemma3}
To obtain~\eqref{eqn:rootfinding}, we first apply Tailor series approximation $\left(1-\frac{hp_i}{l_iM_i}\right)^{n_i-1}\approx e^{-h\frac{n_ip_i}{Ml_i}}$ to the objective function given by Eqn.~\eqref{eq:exp_S}, then substitute $x=\frac{n_ip_i}{Ml_i}$, and finally simplify the sum as a partial sum of a geometric series, obtaining:
\begin{equation}
S = xM \sum_{h=1}^{l_i} e^{-hx}=xMe^{-x}\frac{1-e^{-xl_i}}{1-e^{-x}}.
\label{eqn:thr_approx}
\end{equation}
Instead of directly optimizing the expression, we apply logarithmic transformation to~\eqref{eqn:thr_approx}, and then obtain its derivative
\begin{eqnarray}
&&\frac{dS}{dx}= \nonumber\\
&&\frac{d\left(\log xM + \log e^{-x} + \log\left(1-e^{-xl_i}\right) - \log\left(1-e^{-x}\right)\right)}{dx}.\nonumber
\end{eqnarray}
By simplifying the equation and setting it to $0$, we obtain~\eqref{eqn:rootfinding}. Uniqueness of the root follows from Lemma~\ref{lemma:fixed_k_problem}.

\balance

\end{document}